\documentclass[12pt]{article}

\usepackage{epsfig}
\usepackage{fullpage}
\usepackage{latexsym}
\usepackage{amsthm}
\usepackage{amsfonts}
\usepackage{xspace}
\usepackage{subfigure,xspace}

\usepackage{multirow,booktabs}

\newtheorem{theorem}{Theorem}

\newtheorem{lemma}[theorem]{Lemma}
\newtheorem{property}[theorem]{Property}

\newtheorem{claim}[theorem]{Claim}

\newcommand{\desc}{\ensuremath{\mathrm{desc}}}
\newcommand{\poly}{\ensuremath{\mathrm{poly}}}
\newcommand{\OPT}{\ensuremath{\mathrm{OPT}\xspace}}
\newcommand{\BMC}{{BMC}\xspace}
\newcommand{\DST}{{DST}\xspace}

\begin{document}

\title{The Knapsack Problem with Neighbour Constraints}

\author{Glencora Borradaile\footnote{Supported
by NSF grant CCF-0963921.}\\Oregon State University\\{\tt
glencora@eecs.oregonstate.edu} \and Brent
Heeringa\footnote{Supported by NSF grant IIS-0812514}\\Williams College\\{\tt heeringa@cs.williams.edu} \and
Gordon Wilfong\\Bell Labs\\ {\tt
gtw@research.bell-labs.com}}

\maketitle

\begin{abstract}
We study a constrained version of the knapsack problem in which
dependencies between items are given by the adjacencies of a graph.  In the {\em 1-neighbour knapsack problem}, an item can be selected only if at least one of its neighbours is also
selected.  In the {\em all-neighbours knapsack problem}, an item can be selected only if all its neighbours are also selected.

We give approximation algorithms and
hardness results when the nodes have both uniform and arbitrary
weight and profit functions, and when the dependency graph is
directed and undirected.
\end{abstract}

\newpage

\section{Introduction}

We consider the knapsack problem in the presence of constraints.  The input is a graph $G = (V,E)$ where each vertex $v$ has a {\em weight} $w(v)$ and a {\em profit} $p(v)$, and a knapsack of size $k$.  We start with the usual knapsack goal---find a set of vertices of maximum profit whose total weight does not exceed $k$---but consider two natural variations.  In the {\em 1-neighbour knapsack problem}, a vertex can be selected only if {\em at least one} of its neighbours is also selected (vertices with no neighbours can always be selected).  In the {\em all-neighbour knapsack problem} a vertex can be selected only if all its neighbours are also selected.

We consider the problem with {\em general}
(arbitrary) and {\em uniform} ($p(v) = w(v) = 1\ \forall v$) weights
and profits, and with undirected and directed graphs.  In the case of
directed graphs, the constraints only apply to the {\em out}-neighbours of a vertex.

Constrained knapsack problems have applications to scheduling, tool
management, investment strategies and database
storage~\cite{KPP,BFFS05,Johnson:1983p1256}.  There are also
applications to network formation.  For example, suppose a set of
customers $C \subset V$ in a network $G = (V,E)$ wish to connect to a
server, represented by a single sink $s \in V$.  The server may
activate each edge at a cost and each customer would result in a
certain profit.  The server wishes to activate a subset of the edges
with cost within the server's budget. By introducing a vertex mid-edge
with zero-profit and weight equal to the cost of the edge and giving
each customer zero-weight, we convert this problem to a 1-neighbour
knapsack problem.

\subsection{Results}

We show that the eight resulting problems
\[
\{\mbox{1-neighbour, all-neighbours}\} \times \{\mbox{general, uniform}\}\times\{\mbox{undirected, directed}\}
\]
vary in complexity but afford several algorithmic approaches.  We
summarize our results for the 1-neighbour knapsack problem in Table~\ref{tbl:results}.  In addition, we show that uniform, directed
all-neighbour knapsack has a PTAS but is NP-complete.  The general,
undirected all-neighbour knapsack problem reduces to 0-1 knapsack, so there is a fully-polynomial time approximation scheme.

\begin{table}[tb]
\begin{center}
\begin{tabular}{cccc}
\toprule

&  & Upper & Lower \\
\cmidrule(r){2-4}
\multirow{2}{*}{\hspace{3mm}Uniform\hspace{3mm}} & \hspace{3mm}Undirected\hspace{3mm} & \multicolumn{2}{c}{linear-time exact}  \\
\cmidrule(r){2-4}
& Directed & PTAS & \hspace{3mm}NP-hard (strong sense) \hspace{3mm} \\
\cmidrule(r){2-4}
\multirow{2}{*}{General} & Undirected & \hspace{3mm}$\frac{(1-\varepsilon)}{2} \cdot (1-1/e^{1-\varepsilon})$\hspace{3mm} &   $1-1/e+\epsilon$ \\
\cmidrule(r){2-4}
& Directed & {\em open} & $1/ \Omega(\log^{1-\varepsilon} n)$ \\
\bottomrule
\end{tabular}
\end{center}
\caption{\label{tbl:results} 1-Neighbour Knapsack Problem results:  upper and lower bounds on the approximation ratios for combinations of $\{\mbox{general, uniform}\} \times \{\mbox{undirected, directed}\}$.  For uniform, undirected, the bounds are running-times of optimal algorithms.}

\end{table}

In Section~\ref{sec:g1n} we describe a greedy algorithm that applies to the general 1-neighbour problem for both directed and undirected dependency graphs.  The algorithm requires two oracles: one for finding a set of vertices with high profit and another for finding a set of vertices with high profit-to-weight ratio.  In both cases, the total weight of the set cannot exceed the knapsack capacity and the subgraph defined by the vertices must adhere to a strict combinatorial structure which we define later.   The algorithm achieves an approximation ratio of $(\alpha/2) \cdot (1-1/e^{\beta})$.  The approximation ratios of the oracles determines the $\alpha$ and $\beta$ terms respectively.

For the general, undirected 1-neighbour case, we give polynomial-time oracles that achieve $\alpha = \beta = (1-\varepsilon)$ for any $\varepsilon > 0$.  This yields a polynomial time $((1-\varepsilon)/2) \cdot (1-1/e^{1-\varepsilon})$-approximation.  We also show that no approximation ratio better than $1-1/e$ is possible (assuming P$\not=$NP).  This matches the upper bound up to (almost) a factor of 2.  These results appear in Section~\ref{sec:gu1n}.

In Section~\ref{sec:gd1n}, we show that the general, directed 1-neighbour knapsack problem is $1/ \Omega(\log^{1-\varepsilon} n)$-hard to approximate, even in DAGs.  

In Section~\ref{sec:ud1n} we show that the uniform, directed
1-neighbour knapsack problem is NP-hard in the strong sense but that it has a polynomial-time
approximation scheme (PTAS)\footnote{A PTAS is an algorithm that,
given a fixed constant $\varepsilon < 1$, runs in polynomial time and
returns a solution within $1-\varepsilon$ of optimal.  The algorithm
may be exponential in $1/\varepsilon$}.  Thus, as with general, undirected 1-neighbour problem, our upper and lower bounds are essentially matching.

In Section~\ref{sec:uu1n} we show that the uniform, undirected
1-neighbour knapsack problem affords a simple, linear-time solution.

In Section~\ref{sec:all-neighbours} we show that uniform, directed
all-neighbour knapsack has a PTAS but is NP-complete.  We also discuss the general, undirected all-neighbour problem.

\subsection{Related work} \label{sec:related}

There is a tremendous amount of work on maximizing submodular functions under a single knapsack constraint~\cite{Sviridenko:orl2004}, multiple knapsack constraints~\cite{Kulik:2009}, and both knapsack and matroid constraints~\cite{Lee:2009,groundan-schulz:prepreint2009}.  While our profit function is submodular, the constraints given by the graph are not characterized by a matroid (our solutions, for example, are not closed downward).  Thus, the 1-neighbour knapsack problem represents a class of knapsack problems with realistic constraints that are not captured by previous work.

As we show in Section~\ref{sec:apx-hardness}, the general, undirected 1-neighbour knapsack problem generalizes several maximum coverage problems including the budgeted variant considered by Khuller, Moss, and Naor~\cite{kmn:ipl1999} which has a tight $(1-1/e)$-approximation unless P=NP.  Our algorithm for the general 1-neighbour problem follows the approach taken by Khuller, Moss, and Naor but, because of the dependency graph, requires several new technical ideas.  In particular, our analysis of the greedy step represents a non-trivial generalization of the standard greedy algorithm for submodular maximization.

Johnson and Niemi~\cite{Johnson:1983p1256} give an FPTAS for knapsack
problems on dependency graphs that are in-arborescences (these are
directed trees in which every arc is directed toward a single root).
In their problem formulation, the constraints are given as
out-arborescences---directed trees in which every arc is directed away
from a single root---and feasible solutions are subsets of vertices
that are closed under the {\em predecessor} operation. This problem
can be viewed as an instance of the general, directed 1-neighbour
knapsack problem.

In the subset-union knapsack problem (SUKP)~\cite{KPP}, each item is a subset
of a ground set of elements.  Each element in the ground set has a
weight and each item has a profit and the goal is to find a
maximum-profit set of elements where the weight of the union of the
elements in the sets fits in the knapsack.  It is easy to see that
this is a special case of the general, directed all-neighbours
knapsack problem in which there is a vertex for each item and each
element and an arc from an item to each element in the item's set.
In~\cite{KPP}, Kellerer, Pferschy, and Pisinger show that SUKP is
NP-hard and give an optimal but badly exponential algorithm.  The
precedence constrained knapsack problem~\cite{BFFS05} and
partially-ordered knapsack problem~\cite{Kolliopoulos:2007p1242} are
special cases of the general, directed all-neighbours knapsack problem
in which the dependency graph is a DAG.  Hajiaghayi et.~al.~show that
the partially-ordered knapsack problem is hard to approximate within a
$2^{\log^\delta n}$ factor unless
3SAT$\in$DTIME$(2^{n^{3/4+\varepsilon}})$~\cite{Hajiaghayi:2006p1244}.

\subsection{Notation.}

We consider graphs $G$ with $n$ vertices $V(G)$ and $m$ edges $E(G)$.
Whether the graph is directed or undirected will be clear from
context and we refer to edges of directed graphs as arcs.
For an undirected graph, $N_G(v)$ denotes the neighbours of a vertex $v$ in
$G$.  For a directed graph, $N_G(v)$ denotes the out-neighbours of $v$
in $G$, or, more formally, $N_G(v) = \{u : vu \in E(G)\}$.  Given a set of nodes $X$, $N^{-}_{G}(X)$ is the set of nodes not in $X$ but that have a neighbour (or out-neighbour in the directed case) in $X$.  That is, $N^{-}_{G}(X)=\{u : uv \in E(G), u \not\in X, \mbox{ and } v \in X\}$.
The degree (in undirected graphs) and out-degree (in directed graphs) of a vertex $v$
in $G$ is denoted $\delta_G(v)$.  The subscript $G$ will be dropped
when the graph is clear from context.  For a set of vertices {\em or}
edges $U$, $G[U]$ is the graph induced on $U$.

For a directed graph $G$, $\mathcal{D}$ is the directed, acyclic graph (DAG)
resulting from contracting maximal strongly-connected components
(SCCs) of $G$.  For each node $u \in V(\mathcal{D})$, let $V(u)$ be the set of
vertices of $G$ that are contracted to obtain $u$.

For a vertex $u$,
let $\desc_G(u)$ be the set of all descendants of $u$ in $G$, {\em i.e.},~all
the vertices in $G$ that are reachable from $u$ (including $u$).  A
vertex is its own descendant, but not its own
strict descendant.

For convenience, extend any function $f$ defined on items in a set $X$
to any subset $A \subseteq X$ by letting $f(A) = \sum_{a \in A} f(a)$.
If $f(a)$ is a set, then $f(A) = \bigcup_{a\in A} f(a)$.  If $f$ is
defined over vertices, then we extend it to edges: $f(E) =
f(V(E))$.  For any knapsack problem, \OPT~is the set of
vertices/items in an optimal solution.

\subsection{Viable Families and Viable Sets.}

A set of nodes $U$ is a {\em 1-neighbour set} for $G$ if for every vertex $v \in U$, $|N_{G[U]}(v)| \geq \min\{\delta_{G}(v),1\}$.  That is, a 1-neighbour set is feasible with respect to the dependency graph.  A family of graphs $\mathcal{H}$ is a {\em viable family} for $G$ if, for any subgraph $G'$ of $G$, there exists a partition $\mathcal{Y}_{\mathcal{H}}(G')$ of $G'$ into 1-neighbour sets for $G'$, such that for every $Y \in \mathcal{Y}_{\mathcal{H}}(G')$, there is a graph $H \in \mathcal{H}$ spanning $G[Y]$.  For directed graphs, we take {\em spanning} to mean that $H$ is a directed subgraph of $G[Y]$ and that $Y$ and $H$ contain the same number of nodes.  For a graph $G$,  we call $\mathcal{Y}_{\mathcal{H}}(G)$ a {\em viable partition} of $G$ with respect to $\mathcal{H}$.

\begin{figure}[bt]
\centering\includegraphics[scale=0.5]{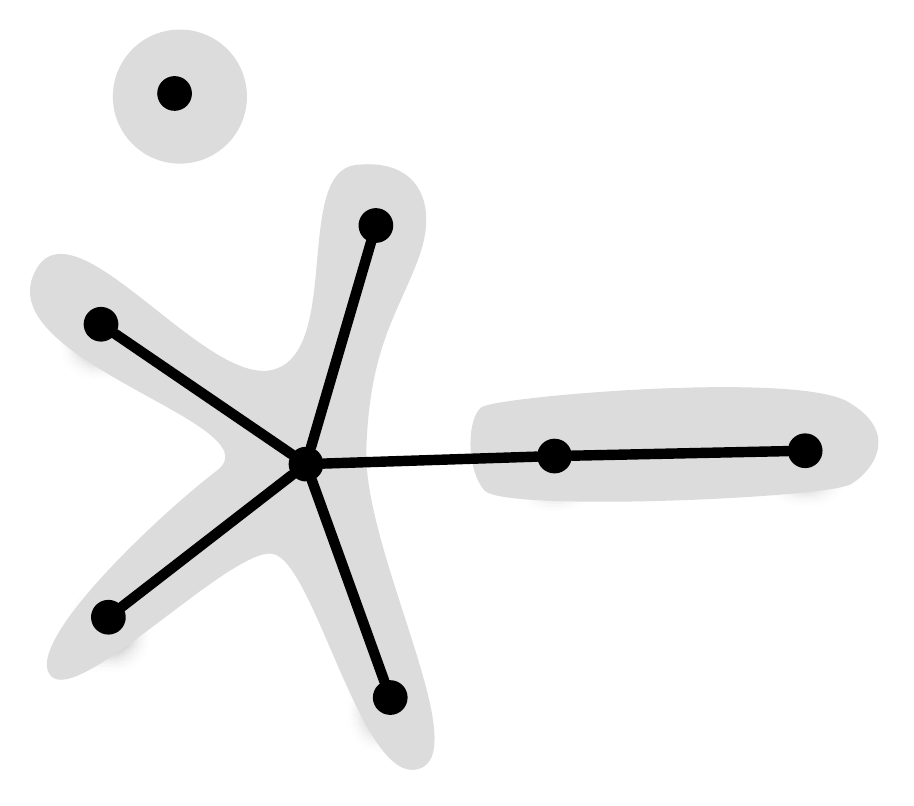}
\caption{ \label{fig:viable-partition}  An undirected graph.   If $\mathcal{H}$ is the family of star graphs, then the shaded regions give the only viable partition of the nodes---no other partition yields 1-neighbour sets.  However, every {\em edge} viable with respect to $\mathcal{H}$.  The singleton node is also viable since it is a 1-neighbour set for the graph.}
\end{figure}

\begin{figure}[tb]
\centering
\subfigure[]{\includegraphics[scale=0.5]{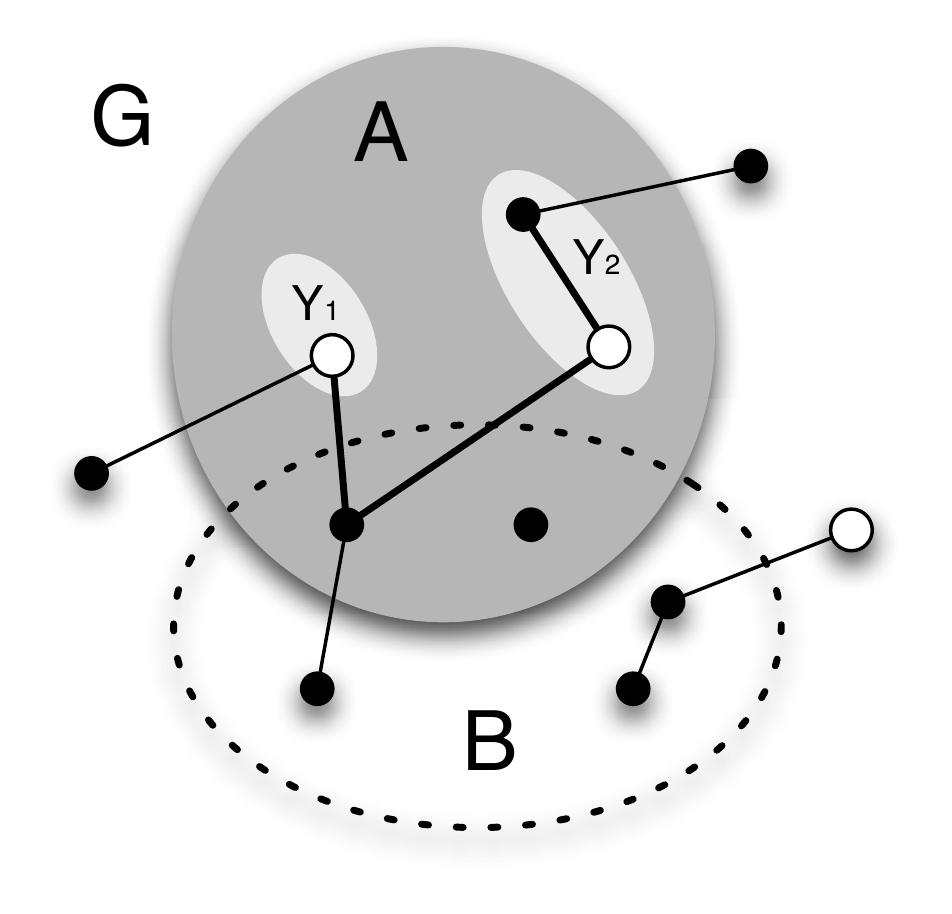}}
\subfigure[]{\includegraphics[scale=0.5]{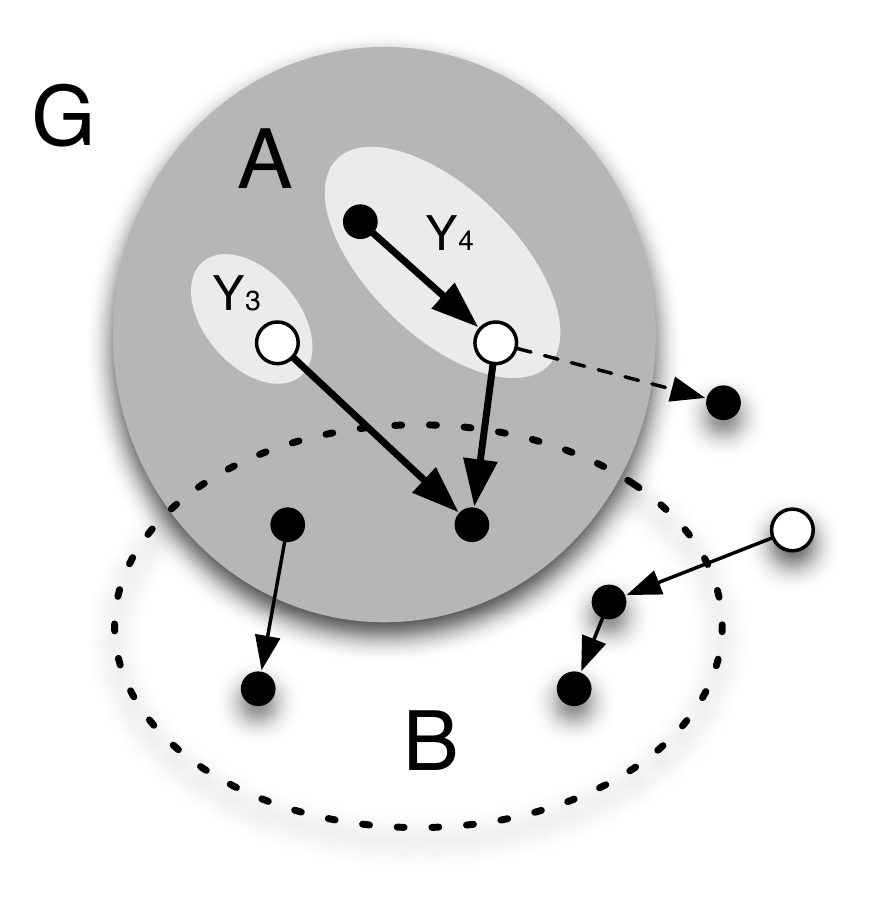}}
\caption{ \label{fig:viable-partition-lemma}  An undirected $G$ in (a) and a directed graph $G$ in (b) with 1-neighbour sets $A$ (dark shaded) and $B$ (dotted) marked in both.  Similarly, in both (a) and (b) the lightly shaded regions give viable partitions for $G[A \setminus B]$ and the white nodes denote $N^{-}_{G}(B)$.  In (a) $Y_{2}$ is viable for $G[A \setminus B]$, and since $|Y_{2}|=2$, it is viable for $G[V(G) \setminus B]$.  $Y_{1}$ is not viable for $G[V(G) \setminus B]$ but it is in $N^{-}_{G}(B)$.  In (b), $Y_{3}$ is viable in $G[V(G) \setminus B]$ whereas $Y_{4}$ is a viable because we consider $G[V(G) \setminus B]$ with the dotted arc removed.}
\end{figure}

In Section~\ref{sec:gu1n} we show that star graphs form a viable
family for any undirected dependency graph.  That is, we show that any
undirected graph can be partitioned into 1-neighbour sets that are
stars.  Fig.~\ref{fig:viable-partition} gives an example.  In
contrast, edges do not form a viable family since, for example, a
simple path with 3 nodes cannot be partitioned into 1-neighbour sets
that are edges.  For DAGs, in-arborescences are a viable family  but
directed paths are not  (consider a directed graph with 3 nodes
$u,v,w$ and two arcs $(u,v)$ and $(w,v)$).  Note that every vertex
must be included as a set on its own in any viable family.

A 1-neighbour set $U$ for $G$ is {\em viable} with respect to $\mathcal{H}$ if there is a graph $H \in \mathcal{H}$ spanning $G[U]$. Note that the 1-neighbour sets in $\mathcal{Y}_{\mathcal{H}}(G)$ are, by definition, viable for $G$, but a viable set for $G$ need not be in $\mathcal{Y}_{\mathcal{H}}(G)$.  For example, if $\mathcal{H}$ is the family of stars and $G$ is the undirected graph in Fig.~\ref{fig:viable-partition}, then any edge is a viable set for $G$ but the only viable partition is the shaded region.  Note that if $U$ is a viable set for $G$ then it is also a viable set for any subgraph $G'$ of $G$ provided $U \subseteq V(G')$.

Viable families and viable sets play an essential role in our greedy algorithm for the general 1-neighbour knapsack problem.  Viable families establish a set of structures over which our oracles can search.  This restriction simplifies both the design and analysis of efficient oracles as well as coupling the oracles to a shared family of graphs which, as we'll show later, is essential to our analysis.  In essence, viable families provide a mechanism to coordinate the oracles into returning sets with roughly similar structure.  Viable sets correctly capture the idea of an indivisible unit of choice in the greedy step.  We formalize this with the following lemma which is illustrated in Fig.~\ref{fig:viable-partition-lemma}.

\begin{lemma} \label{lemma:viable-correct} Let $G$ be a graph and
$\mathcal{H}$ be a viable family for $G$.  Let $A$ and $B$ be
1-neighbour sets for $G$.  If $\mathcal{Y}_{\mathcal{H}}(C)$ is a
viable partition of $G[C]$ where $C=A \setminus B$ then every set $Y
\in \mathcal{Y}_{\mathcal{H}}(C)$ is either (i) a singleton node $y$
such that $y \in N^{-}_{G}(B)$ (i.e., $y$ has a neighbour in $B$),
or (ii) a viable set for $G'$, which is the subgraph obtained by deleting
vertices in $B$ and arcs in $X$ where $X$ is empty if $G$ is
undirected and $X$ is the set of arcs with tails in $N^{-}_{G}(B)$
if $G$ is directed.

\end{lemma}

\begin{proof}
If $|Y|=1$ then let $Y=\{y\}$.  If $\delta_{G}(y)=0$ then $Y$ is a
viable set for $G$ so it is viable set for $G'$. Otherwise, since
$A$ is a 1-neighbour set for $G$, $y$ must have a neighbour in $B$
so $y \in N^{-}_{G}(B)$.  If $|Y| >1$ then, provided $G$ is
undirected, $Y$ is also a viable set in $G$ so it is a viable set
in $G'$.  If $G$ is directed and $Y$ contains a node $y$ that is in
$N^{-}_{G}(B)$, an arc out of $y$ is not needed for feasibility
since $y$ already has a neighbour in $A$.

\end{proof}

\section{The general 1-neighbour knapsack problem} \label{sec:g1n}

Here we give a greedy algorithm {\sc Greedy-1-Neighbour} for the general 1-neighbour knapsack problem on both directed and undirected graphs.  A formal description of our algorithm is available in Fig.~\ref{alg:greedy-1-neighbour}.  {\sc Greedy1-Neighbour} relies on two oracles {\sc Best-Profit-Viable} and {\sc Best-Ratio-Viable} which find viable sets of nodes with respect to a fixed viable family $\mathcal{H}$.  In each iteration $i$, we call {\sc Best-Ratio-Viable} which, given the nodes not yet chosen by the algorithm, returns the highest profit-to-weight ratio, viable set $S_{i}$ with weight not exceeding the remaining capacity.  We also consider the set of nodes $Z$ not in the knapsack, but with at least one neighbour already in the knapsack.  Let $s_{i}$ be the node with highest profit-to-weight ratio in $Z$ not exceeding the remaining capacity.  We greedily add either $s_{i}$ or $S_{i}$ to our knapsack $U$ depending on which has higher profit-to-weight ratio. We continue until we can no longer add nodes to the knapsack.

For a viable family $\mathcal{H}$, if we can efficiently approximate the highest profit-to-weight ratio viable set to within a factor of $\beta$ and if we can efficiently approximate the highest profit viable set to within a factor of $\alpha$, then our greedy algorithm yields a polynomial time $\frac{\alpha}{2}(1-1/e^\beta)$-approximation.

\begin{figure}[tb]
\begin{center}
\fbox{
\begin{minipage}[h]{.9\linewidth}
\noindent {\sc Greedy-1-Neighbour}$(G,k):$
\begin{tabbing}
\qquad \= $S_{\max}$ = {\sc best-profit-viable}$(G,k)$ \\
\> $K = k$, $U = \emptyset$, $i = 1$, $G' = G$, $Z = \emptyset$ \\
\> WHILE there is either a viable set in $G'$ or a node in $Z$ with weight $\leq K$ \\
\> \qquad \= $S_i$ = {\sc best-ratio-viable}$(G', K)$ \\
\> \> $s_{i} = \arg\max \{p(v) / w(v) \,|\, v \in Z\}$ \\
\> \> IF $p(s_{i}) / w(s_{i}) \, > \, p(S_i) / w(S_i)$\\
\> \> \qquad $S_{i} = \{ s_{i} \}$\\
\> \> $G' = G[V(G') \setminus S_{i}]$\\
\> \> $i = i+1,\ U = U \cup V(S_i), \ K = K-w(S_i)$\\
\> \> $Z = N^{-}_{G}(U)$\\
\> \> If $G$ is directed, remove any arc in $G'$ with a tail in $Z$\\
\> RETURN $\arg\max \{p(S_{\max}), p(U)\}$
\end{tabbing}
\end{minipage}
}
\end{center}
\caption{\label{alg:greedy-1-neighbour} The {\sc Greedy-1-Neighbour} algorithm.  In each iteration $i$, we greedily add either the viable set $S_{i}$ or the node $s_{i}$ to our knapsack $U$ depending on which has higher profit-to-weight ratio.  This continues until we can no longer add nodes to the knapsack.}
\vspace{-4mm}
\end{figure}

\begin{theorem} \label{thm:gd1n}
{\sc Greedy-1-Neighbour} is a
$\frac{\alpha}{2}(1-\frac{1}{e^\beta})$-approximation for the
general 1-neighbour problem on directed and undirected graphs.
\end{theorem}

\begin{proof}
Let $\OPT$ be the set of vertices in an optimal solution.  In addition, let $U_i =\cup_{j = 1}^{i} V(S_j)$ correspond to $U$ after the first $i$ iterations where $U_{0} = \emptyset$.  Let $\ell+1$ be the first iteration in which there is either a node in $Z \cap \OPT$ or a viable set in $\OPT \setminus U_\ell$ whose profit-to-weight ratio is larger than $S_{\ell+1}$.  Of these, let $\mathcal{S}_{\ell+1}$ be the node or set with highest profit-per-weight.  For convenience, let  $\mathcal{S}_{i} = S_{i}$ and $\mathcal{U}_{i} = U_{i}$ for $i = 1 \ldots \ell$, and $\mathcal{U}_{\ell+1} = \mathcal{U}_{\ell} \cup \mathcal{S}_{\ell+1}$.  Notice that $\mathcal{U}_{\ell}$ is a feasible solution to our problem but that $\mathcal{U}_{\ell+1}$ is not since it contains $\mathcal{S}_{\ell+1}$ which has weight exceeding $K$.  We analyze our algorithm with respect to $\mathcal{U}_{\ell+1}$.

\begin{lemma} \label{lem:profit-inc-dag}
For each iteration $i = 1, \ldots, \ell+1$, the following holds:
$$p(\mathcal{S}_i) \geq \beta\frac{w(\mathcal{S}_i)}{k}\left(p(\OPT)-p(\mathcal{U}_{i-1})\right)$$
\end{lemma}

\begin{proof}
Fix an iteration $i$ and let $I$ be the graph induced by $\OPT \setminus \mathcal{U}_{i-1}$.  Since both $\OPT$ and $\mathcal{U}_{i-1}$ are 1-neighbour sets for $G$, by Lemma~\ref{lemma:viable-correct}, each $Y \in \mathcal{Y}_{\mathcal{H}}(I)$ is either a viable set for $G'$ (so it can be selected by {\sc best-ratio-viable}) or a singleton vertex in $N^{-}_{G}(\mathcal{U}_{i-1})$ (which {\sc Greedy-1-Neighbour} always considers).  Thus, if $i \leq \ell$, then by the greedy choice of the algorithm and approximation ratio of {\sc best-ratio-viable} we have
\begin{equation} \label{eq:viable-lb}
\frac{p(\mathcal{S}_i)}{w(\mathcal{S}_i)} \geq \beta \frac{p(Y)}{w(Y)} \ \mbox{for all} \ Y \in \mathcal{Y}_{\mathcal{H}}(I).
\end{equation}
If $i=\ell+1$ then $p(\mathcal{S}_{\ell+1})/w(\mathcal{S}_{\ell+1})$ is, by definition, at least as large as the profit-to-weight ratio of any $Y \in \mathcal{Y}$.  It follows that for $i=1, \ldots, \ell+1$:
\begin{eqnarray*}
p(\OPT) - p(\mathcal{U}_{i-1}) = \sum_{Y \in \mathcal{Y}_{\mathcal{H}}(I)} p(Y) & \leq &
\frac{1}{\beta}\frac{p(\mathcal{S}_i)}{w(\mathcal{S}_i)} \sum_{Y \in \mathcal{Y}_{\mathcal{H}}(I)}
w(Y),\mbox{ by Eq.~(\ref{eq:viable-lb}) }\\
& \leq & \frac{1}{\beta}\frac{p(\mathcal{S}_i)}{w(\mathcal{S}_i)} w(\OPT),\mbox{
since $I$ is a subset of \OPT} \\
& \leq & \frac{1}{\beta}\frac{k}{w(\mathcal{S}_i)}p(\mathcal{S}_i),\mbox{ since
$w(\OPT) \leq k$}
\end{eqnarray*}
Rearranging gives Lemma~\ref{lem:profit-inc-dag}.
\hfill  \end{proof}

\begin{lemma} \label{lem:profit}
For $i = 1, \ldots, \ell+1$, the following holds:
$$p(\mathcal{U}_i) \geq \left[1-\prod_{j = 1}^i
\left(1-\beta\frac{w(\mathcal{S}_j)}{k} \right) \right] p(\OPT)$$
\end{lemma}

\begin{proof}
We prove the lemma by induction on $i$.  For $i = 1$, we need to
show that
\begin{equation}\label{eq:induct-1}
p(\mathcal{U}_1)  \geq \beta\frac{w(\mathcal{S}_1)}{k}
p(\OPT).
\end{equation}
This follows immediately from Lemma~\ref{lem:profit-inc-dag} since $p(\mathcal{U}_{0})=0$ and $\mathcal{U}_{1}=\mathcal{S}_1$. Suppose the lemma holds for iterations 1 through $i-1$.  Then it is
easy to show that the inequality holds for iteration $i$ by applying
Lemma~\ref{lem:profit-inc-dag} and the inductive hypothesis.  This
completes the proof of Lemma~\ref{lem:profit}.
\hfill
\end{proof}

We are now ready to prove Theorem~\ref{thm:gd1n}.  Starting with the inequality in Lemma~\ref{lem:profit} and using
the fact that adding $\mathcal{S}_{\ell+1}$ violates the knapsack constraint (so
$w(\mathcal{U}_{\ell+1}) > k$) we have
\begin{eqnarray*}
p(\mathcal{U}_{\ell+1}) & \geq & \left[1-\prod_{j = 1}^{\ell + 1}
\left(1-\beta\frac{w(\mathcal{S}_j)}{k} \right) \right] p(\OPT) \\
& \geq &\left[1-\prod_{j = 1}^{\ell + 1}
\left(1-\beta\frac{w(\mathcal{S}_j)}{w(U_{\ell+1})} \right) \right]
p(\OPT) \\
& \geq & \left[1-\left(1-\frac{\beta}{\ell+1}\right)^{\ell+1} \right]
p(\OPT) \geq  \left(1-\frac{1}{e^\beta}\right)p(\OPT)
\end{eqnarray*}
where the penultimate inequality follows because equal $w(\mathcal{S}_j)$ maximize the product.  Since $S_{\max}$ is within a factor of $\alpha$ of the maximum profit viable set of weight $\leq k$ and $\mathcal{S}_{\ell+1}$ is contained in \OPT,
$p(S_{\max}) \geq \alpha \cdot p(\mathcal{S}_{\ell+1})$.  Thus, we have $p(U) + p(S_{\max})/ \alpha \geq p(\mathcal{U}_\ell)
+ p(\mathcal{S}_{\ell+1}) = p(\mathcal{U}_{\ell+1}) \geq
\left(1-\frac{1}{e^\beta}\right)p(\OPT)$.  Therefore
$\max\{p(U), p(S_{\max})\} \geq
\frac{\alpha}{2}\left(1-\frac{1}{e^\beta}\right)p(\OPT)$.

\hfill
\end{proof}

\subsection{The general, undirected 1-neighbour problem} \label{sec:gu1n}

Here we formally show that stars are a viable family for undirected
graphs and describe polynomial-time implementations of {\sc
Best-Profit-Viable} and {\sc Best-Ratio-Viable} for the star family.  Both oracles achieve an approximation ratio of $(1-\varepsilon)$ for any $\varepsilon > 0$.  Combined with {\sc Greedy-1-Neighbour} this yields a polynomial time $((1-\varepsilon)/2) \cdot (1-1/e^{1-\varepsilon})$-approximation for the general, undirected 1-neighbour problem.  In addition, we show that this approximation is nearly tight by showing that the general, undirected 1-neighbour problem generalizes many coverage problems including the max $k$-cover and budgeted maximum coverage, neither of which have a $(1-1/e+\epsilon)$-approximation for any $\epsilon > 0$ unless P=NP.

\subsubsection{Stars} \label{sec:stars}

For the rest of this section, we assume $\mathcal{H}$ is the family of
star graphs ({\em i.e.} graphs composed of a center vertex $u$ and a (possibly empty) set of edges all of which have $u$ as an endpoint) so that
given a graph $G$ and a capacity $k$, {\sc Best-Profit-Viable} returns
the highest profit, viable star with weight at most $k$ and {\sc Best-Ratio-Viable} returns the highest profit-to-weight, viable star with weight at most $k$.

\begin{lemma}\label{lem:graphs-into-stars}
The nodes of any undirected constraint graph $G$ can be partitioned into 1-neighbour sets that are stars.
\end{lemma}

\begin{proof}
Let $G_{i}$ be an arbitrary connected component of $G$.  If $|V(G_{i})|=1$ then $V(G_{i})$ is trivially a 1-neighbour set and the trivial star consisting of a single node is a spanning subgraph of $G_{i}$.  If $G_{i}$ is non-trivial then let $T$ be any spanning tree of $G_{i}$ and consider the following construction: while $T$ contains a path $P$ with $|P| > 2$, remove an interior edge of $P$ from $T$.  When the algorithm finishes, each path has at least one edge and at most two edges, so $T$ is a set of non-trivial stars, each of which is a 1-neighbour set.
\hfill  \end{proof}

\paragraph{{\sc Best-Profit-Viable}}

Finding the maximum profit, viable star of a graph $G$ subject to a knapsack constraint $k$ reduces to the traditional unconstrained knapsack problem which has a well-known FPTAS that runs in $O(n^{3} / \varepsilon)$ time~\cite{ibarra-kim:jacm1975,vazirani}. Every vertex $v \in V(G)$ defines a knapsack problem:  the items are $N_{G}(v)$ and the capacity is $k-w(v)$.  Combining $v$ with the solution returned by the  FPTAS  yields a candidate star.  We consider the candidate star for each vertex and return the one with highest profit.  Since we consider all possible star centers, {\sc Best-Profit-Viable} runs in $O(n^{4} / \varepsilon)$ time and returns a viable star within a factor of $(1-\varepsilon)$ of optimal, for any $\varepsilon > 0$.

\paragraph{{\sc Best-Ratio-Viable}}

We again turn to the FPTAS for the standard knapsack problem.  Our goal is to find a high profit-to-weight star in $G$ with weight at most $k$.  The standard FPTAS for the unconstrained knapsack problem builds a dynamic programing table $T$  with $n$ rows and $nP'$ columns where $n$ is the number of available items and $P'$ is the maximum adjusted profit over all the items.  Given an item $v$, its adjusted profit is $p'(v) = \lfloor \frac{p(v)}{ (\varepsilon / n) \cdot P} \rfloor$ where $P$ is the true maximum profit over all the items.  Each entry $T[i,p]$ gives the weight of the minimum weight subset over the first $i$ items achieving profit $p$.

Notice that, for any fixed profit $p$, $p / T[n,p]$ is the highest profit-to-weight ratio for that $p$.  Therefore, for $1 \leq p \leq nP'$, the $p$ maximizing $p / T[n,p]$ gives the highest profit-to-weight ratio of any feasible subset provided $T[n,p] \leq k$.  Let $S$ be this subset.  We will show that $p(S) / w(S)$ is within a factor of $(1-\varepsilon)$ of \OPT\ where \OPT\ is the profit-to-weight ratio of the highest profit-to-weight ratio feasible subset $S^{*}$.

Letting $r(v) = p(v) / w(v)$ and $r'(v) = p'(v) / w(v)$, and following~\cite{vazirani}, we have
\[
r(S^{*}) - ((\varepsilon / n) \cdot P)  \cdot  r'(S^{*}) \leq \varepsilon P / w(S^{*})
\]
since, for any item $v$, the difference between $p(v)$ and $((\varepsilon / n) \cdot P) \cdot p'(v)$ is at most $(\varepsilon / n) \cdot P$ and we can fit at most $n$ items in our knapsack.  Because $r'(S) \geq r'(S^{*})$ and \OPT \ is at least $P / w(S^{*})$ we have
\[
r(S)  \geq  (\varepsilon / n) \cdot P \cdot r'(S^{*}) \geq  r(S^{*}) - \varepsilon P / w(S^{*}) \geq   \OPT - \varepsilon \OPT = (1-\varepsilon)\OPT.
\]
Now, just as with {\sc Best-Profit-Viable}, every vertex $v \in V(G)$ defines a knapsack instance where $N_{G}(V)$ is the set of items and $k-w(v)$ is the capacity.  We run the modified FPTAS for knapsack on the instance defined by $v$ and add $v$ to the solution to produce a set of candidate stars.  We return the star with highest profit-to-weight ratio.  Since we consider all possible star centers, {\sc Best-Ratio-Viable} runs in $O(n^{4} / \varepsilon)$ time and returns a viable star within a factor of $(1-\varepsilon)$ of optimal, for any $\varepsilon > 0$.

\paragraph{Justifying Stars}

Besides some isolated vertices, our solution is a set of edges, but
the edges are not necessarily vertex disjoint.  Analyzing our greedy
algorithm in terms of edges risks counting vertices multiple times.
Partitioning into stars allows us to charge increases in the profit
from the greedy step without this risk.  In fact, stars are
essentially the {\em simplest} structure meeting this requirement
which is why we use them as our viable family.

\paragraph{Improving the approximation ratio}

Often this style of greedy algorithm can be augmented with an
``enumeration over triples'' step to improve the ratio of
$(1-\epsilon)(1-{1\over e^\epsilon})$.  However, such an enumeration
would require enumerating over all possible triples of {\em stars} in
our case.  Doing so cannot be done in polynomial time, unless the
graph has bounded degree.

\subsubsection{General, undirected 1-neighbour knapsack is APX-complete} \label{sec:apx-hardness}

Here we show that it is NP-hard to approximate the general, undirected 1-neighbour knapsack problem to within a factor better than $1-1/e+\epsilon$ for any $\epsilon > 0$ via an approximation-preserving reduction from max $k$-cover~\cite{feige:jacm1998}.  An instance of max $k$-cover is a set cover instance $(S,{\mathcal R})$ where $S$ is a ground set of $n$ items and $\mathcal R$ is a collection of subsets of $S$.  The goal is to cover
as many items in $S$ using at most $k$ subsets from $\mathcal R$.

\begin{theorem}
The general, undirected 1-neighbour knapsack problem has no $(1-1/e+\epsilon)$-approximation for any $\epsilon > 0$ unless P$=$NP.
\end{theorem}

\begin{proof}
Given an instance of $(S,{\mathcal R})$ of max $k$-cover, build a bipartite
graph $G=(U \cup V, E)$ where $U$ has a node $u_{i}$ for each $s_i \in
S$ and $V$ has a node $v_{j}$ for each set $R_{j} \in {\mathcal R}$.  Add
the edge $\{u_{i}, v_{j}\}$ to $E$ if and only if $u_{i} \in R_{j}$.
Assign profit $p(u_{i})=1$ and weight $w(u_{i})=0$ for each vertex
$u_{i} \in U$ and profit $p(v_{j})=0$ and weight $w(u_{i})=1$ for each
vertex $v_{j} \in V$.  Since no pair of vertices in $U$ have an edge
and since every vertex in $U$ has no weight, our strategy is to pick
vertices from $V$ and all their neighbours in $U$.  Since every
vertex of $U$ has unit profit, we should choose the $k$ vertices from
$V$ which collectively have the most neighbours.  This is exactly the
max $k$-cover problem.

\end{proof}

The max $k$-cover problem represents a class of {\em budgeted maximum
coverage} (\BMC) problems where the elements in the base set have
unit profit (referred to as weights in~\cite{kmn:ipl1999}) and the
cover sets have unit weight (referred to as costs
in~\cite{kmn:ipl1999}).  In fact, one can use the above reduction to
represent an arbitrary \BMC instance: form the same bipartite graph,
assign the element weights in \BMC as vertex profits in $U$, and finally
assign the covering set costs in \BMC as vertex weights in $V$.

\subsection{General, directed 1-neighbour knapsack is hard to approximate} \label{sec:gd1n}

Here we consider the 1-neighbour knapsack problem where $G$ is directed and has arbitrary profits and weights.  We show via a reduction from {\em directed Steiner tree} (\DST) that the general, directed 1-neighbour problem is hard to approximate within a factor of $1/ \Omega(\log^{1-\varepsilon} n)$.  Our result holds for DAGs.  Because of this negative result, we also don't expect that good approximations exist for either {\sc Best-Profit-Viable} and {\sc Best-Ratio-Viable} for any family of viable graphs.

In the \DST problem on DAGs we are given a DAG $G=(V,E)$ where each arc has an associated cost, a subset of $t$
vertices called {\em terminals} and a root vertex $r \in V$.  The goal
is to find a minimum cost set of arcs that together connect $r$ to all
the terminals ({\em i.e.}, the arcs form an out-arborescence rooted at
$r$).  For all $\varepsilon >0$, \DST admits no
$\log^{2-\varepsilon} n$-approximation algorithm unless $NP\subseteq
ZTIME[n^{\poly\log n}]$~\cite{HK}.  This result holds even for very
simple DAGs such as {\em leveled DAGs} in which $r$ is the only root, $r$ is at level 0,
each arc goes from a vertex at level $i$ to a vertex at level $i+1$,
and there are $O(\log n)$ levels.  We use leveled DAGs in our proof of the following theorem.

\begin{theorem} \label{thm:gd1nlb}
The general, directed 1-neighbour knapsack problem is
$1/\Omega(\log^{1-\varepsilon} n)$-hard to approximate unless $NP\subseteq
ZTIME [n^{\poly\log n}]$.
\end{theorem}

\begin{proof}
Let $D$ be an instance of \DST where the underlying graph $G$ is a leveled DAG with a single root $r$.  Suppose there is a solution to $D$ of cost $C$.

\begin{claim} \label{claim:cover}
If there is an $\alpha$-approximation algorithm for
the general, directed 1-neighbour knapsack problem then a solution
to $D$ with cost $O(\alpha \log t)\times C$ can be found where $t$
is the number of terminals in $D$.
\end{claim}

\begin{proof}
Let $G=(V,A)$ be the DAG in instance $D$.  We modify it to
$G'=(V',A')$ where we split each arc $e\in A$ by placing a dummy
vertex on $e$ with weight equal to the cost of $e$ according to $D$
and profit of 0.  In addition, we also reverse the orientation of each arc.
Finally, all other vertices are given weight 0 and terminals are assigned
a profit of 1 while the non-terminal vertices
of $G$ are given a profit of 0.  We create an instance $N$ of the
general, directed 1-neighbour knapsack problem consisting of $G'$
and budget bound of $C$.  By assumption, there is a solution to $N$
with cost $C$ and profit $t$.  Therefore given $N$, an
$\alpha$-approximation algorithm would produce a set of arcs whose
weight is at most $C$ and includes at least $t/\alpha$ terminals.
That is, it has a profit of at least $t/\alpha$.  Set the weights of
dummy nodes to 0 on the arcs used in the solution. Then for all
terminals included in this solution, set their profit to 0 and
repeat.  Standard set-cover analysis shows that after $O(\alpha \log
t)$ repetitions, each terminal will have been connected to the root
in at least one of the solutions.  Therefore the union of all the
arcs in these solutions has cost at most $O(\alpha \log t)\times C$
and connects all terminals to the root.
\hfill  \end{proof}
Using the above claim, we'll show that if there is an $\alpha$-approximation algorithm for the general, directed-1-neighbour problem then there is an $O(\alpha \log
t)$-approximation algorithm for \DST which implies the theorem.  Let $L$ be the total cost of the arcs in the instance of \DST.  For
each $2^i < L$, take $C=2^i$ and perform the procedure in the previous claim for $\alpha \log t$ iterations.  If after
these iterations all terminals are connected to the root then call
the cost of the resulting arcs a valid cost.  Finally, choose the
smallest valid cost, say $C'$ and $C'$ will be no more than
$2C_{\OPT}$ where $C_{\OPT}$ is the optimal cost of a solution for
the \DST instance.  By the previous claim we have a solution
whose cost is at most $2C_{\OPT} \times O(\alpha \log t)$.
\hfill  \end{proof}

\section{The uniform, directed 1-neighbour knapsack problem} \label{sec:ud1n}

In this section, we give a PTAS for the uniform, directed 1-neighbour
knapsack problem.  We rule out an FPTAS by proving the following theorem.

\begin{theorem} \label{thm:ud1n-hard}
The uniform, directed 1-neighbour problem is strongly NP-hard.
\end{theorem}

\begin{proof}
The proof is a reduction from set cover.  Let the base set for an
instance be $S=\{ s_1, s_2, \ldots, s_{n}\}$ and the collection of
subsets of $S$ be ${\mathcal R}=\{R_1, R_2, \ldots, R_{m}\}$.  The maximum
number of sets desired to cover the base set is $t$.

We build an instance of the 1-neighbour knapsack problem.  Let $M =
n+1$.  The dependency graph is as follows.  For
each subset $R_i$ create a cycle $C_i$ of size $M$; the set of cycles are
pairwise vertex disjoint.  In each such cycle $C_i$ choose some node
arbitrarily and denote it by $c_i$.  For each $s_j\in S$, define a new
node in $V$ and label it $v_j$.  Define $A=\{(v_j,c_i)\; : \; s_j\in
R_i\}$.   Let the capacity of the knapsack be $k = tM+n$.

Suppose ${\mathcal R}'$ is a solution to the set-cover instance.  Since
$1\le |{\mathcal R}'|\le t$, we can define $0\le p <t$ to be such that
$|{\mathcal R}'|+p=t$.  Let ${\mathcal R}''=\{R_{i(1)}, R_{i(2)}, \ldots,
R_{i(p)}\}$ be a collection of $p$ elements of $\mathcal R$ not in ${\mathcal
R}'$. Let $G'$ be the graph induced by the union of the nodes in
$C_j$ for each $R_j\in {\mathcal R}'$ or ${\mathcal R}''$, and $\{v_1, v_2,
\ldots, v_n\}$: $G'$ consists of exactly $tM+n$ nodes.  Every vertex
in the cycles of $G'$ has out-degree 1.  Since ${\mathcal R}'$ is a set
cover, for every $s_j\in S$ there is some $R_i\in {\mathcal R}'$ where
$s_j\in R_i$ and so the arc $(v_j, c_i)$ is in $G'$.  It follows that
$G'$ is a witness for a 1-neighbour set of size $k=tM+n$.

Now suppose that the subgraph $G'$ of $G$ is a solution to the
1-neighbour knapsack instance with value $k$.  Since $M>n$, it is straightforward to
check that $G'$ must consist of a collection $\mathcal C$ of exactly $t$
cycles, say ${\mathcal C}=\{C_{a(1)}, C_{a(2)}, \ldots, C_{a(t)}\}$, and
each node $v_i$, $1\le i\le n$, along with some arc
$(v_i,c_{a(j_i)})$.  But by definition of $G$, that means that $s_i\in
R_{a(j_i)}$ for $1\le i\le n$ and so $\{ R_{a(j_1)},
R_{a(j_2)},\ldots , R_{a(j_n)}\}$ is a solution to the set cover
instance.
\end{proof}

\subsection{A PTAS for the uniform, directed 1-neighbour problem.}
Let $U$ be a 1-neighbour set.  Let $A_U$ be a minimal set of arcs of
$G$ such that for every vertex $u \in U$, $\delta_{G[A_U]}(u) \geq \min \{\delta_G(u),1\}$.  That is, $A_U$ is a
{\em witness} to the feasibility of $U$ as a 1-neighbour set.  Since
each node of $U$ in $G[A_U]$ has out-degree 0 or 1, the structure of
$A_U$ has the following form.

\begin{property}
\label{prop:structure}
Each connected component of $G[A_U]$ is a cycle $C$ and a collection
of vertex-disjoint in-arborescences, each rooted at a node of $C$.  $C$
may be trivial, i.e.,~$C$ may be a single vertex $v$, in which case
$\delta_G(v) = 0$.
\end{property}

For a strongly connected component $X$, let $c(X)$ be the size of the
shortest directed cycle in $X$ with $c(X) = 1$ if and only if $|X| = 1$.

\begin{lemma}
\label{lem:scc-structure}
There is an optimal 1-neighbour knapsack $U$ and a witness $A_U$ such that
for each non-trivial, maximal SCC $K$ of $G$, there is at most one
cycle of $A_U$ in $K$ and this cycle is a smallest cycle of $K$.
\end{lemma}

\begin{proof}
First we modify $A_U$ so that it contains smallest cycles of maximal
SCCs.  We rely heavily on the structure of $A_U$ guaranteed by
Property~\ref{prop:structure}.  The idea is illustrated
in Fig.~\ref{fig:1-neighbour-structure}.

Let $C$ be a cycle of $A_U$ and let $K$ be the maximal SCC of $G$
that contains $C$.  Suppose $C$ is not the smallest cycle of $K$ or
there is more than one cycle of $A_U$ in $K$.  Let $H$ be the
connected component of $A_U$ containing $C$.  Let $C'$ be a smallest
cycle of $K$.  Let $P$ be the shortest directed path from $C$ to
$C'$.  Since $C$ and $C'$ are in a common SCC, $P$ exists.  Let $T$
be an in-arborescence in $G$ spanning $P$, $C$ and $H$ rooted at a
vertex of $C'$.

Some vertices of $C' \cup P$ might already be in the
1-neighbour set $U$: let $X$ be these vertices.  Note that $X$ and $V(H)$ are
disjoint because of Property~\ref{prop:structure}.  Let $T'$ be a
sub-arborescence of $T$ such that:
\begin{itemize}
\item $T'$ has the same root as $T$, and
\item $|V(T' \cup C') \cup X| = |V(H)|+|X|$.
\end{itemize}
Since $|V(T \cup C')| = |V(P\cup H \cup C')| \geq |V(H)| + |X|$ and
$T \cup C'$ is connected, such an in-arborescence exists.

Let $B = (A_U \setminus H) \cup T' \cup C'$.  Let $B'$ be a witness
spanning $V(B)$ contained in $B$ that contains the arcs in $C'$.  We
have that $B'$ has $|U|$ vertices and contains a smallest cycle of
$K$.

We repeat this procedure for any SCC in our witness that contains
a cycle of a maximal SCC of G that is not smallest or contains two
cycles of a maximal SCC.
\hfill  \end{proof}

\begin{figure}[tb]
\centering
\subfigure[]{\includegraphics[scale=0.4]{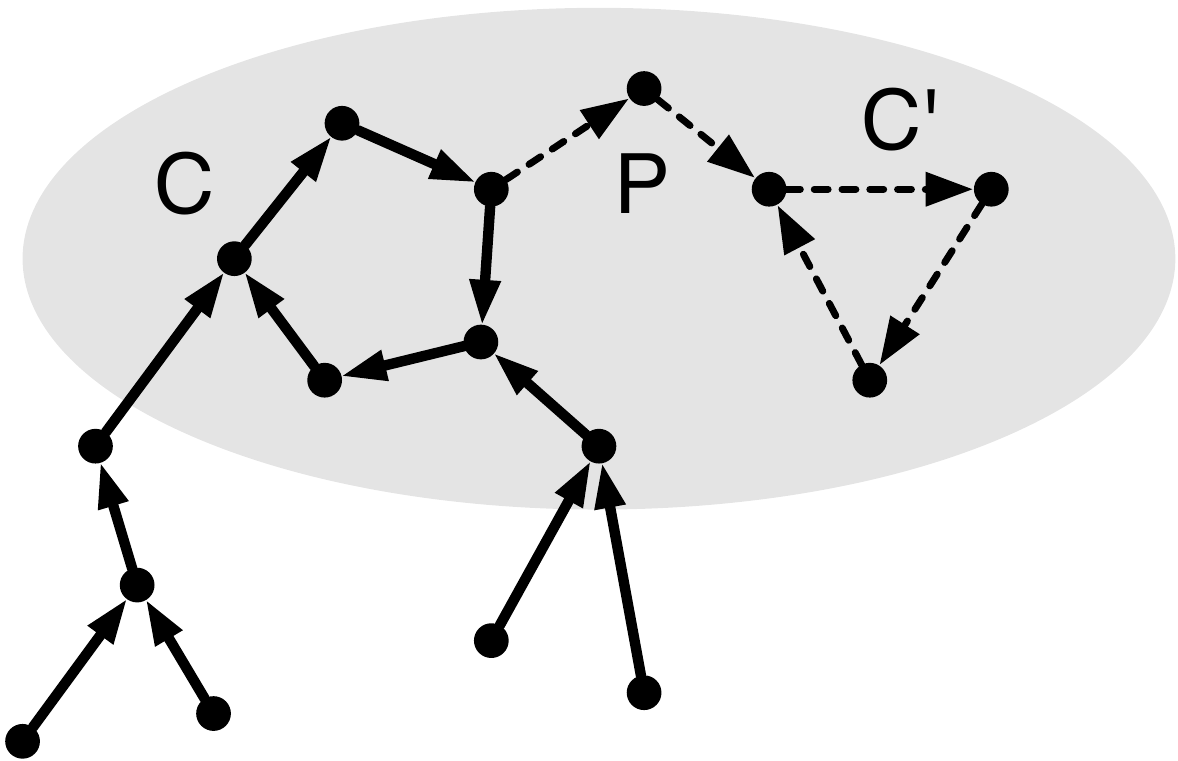}}
\subfigure[]{\includegraphics[scale=0.4]{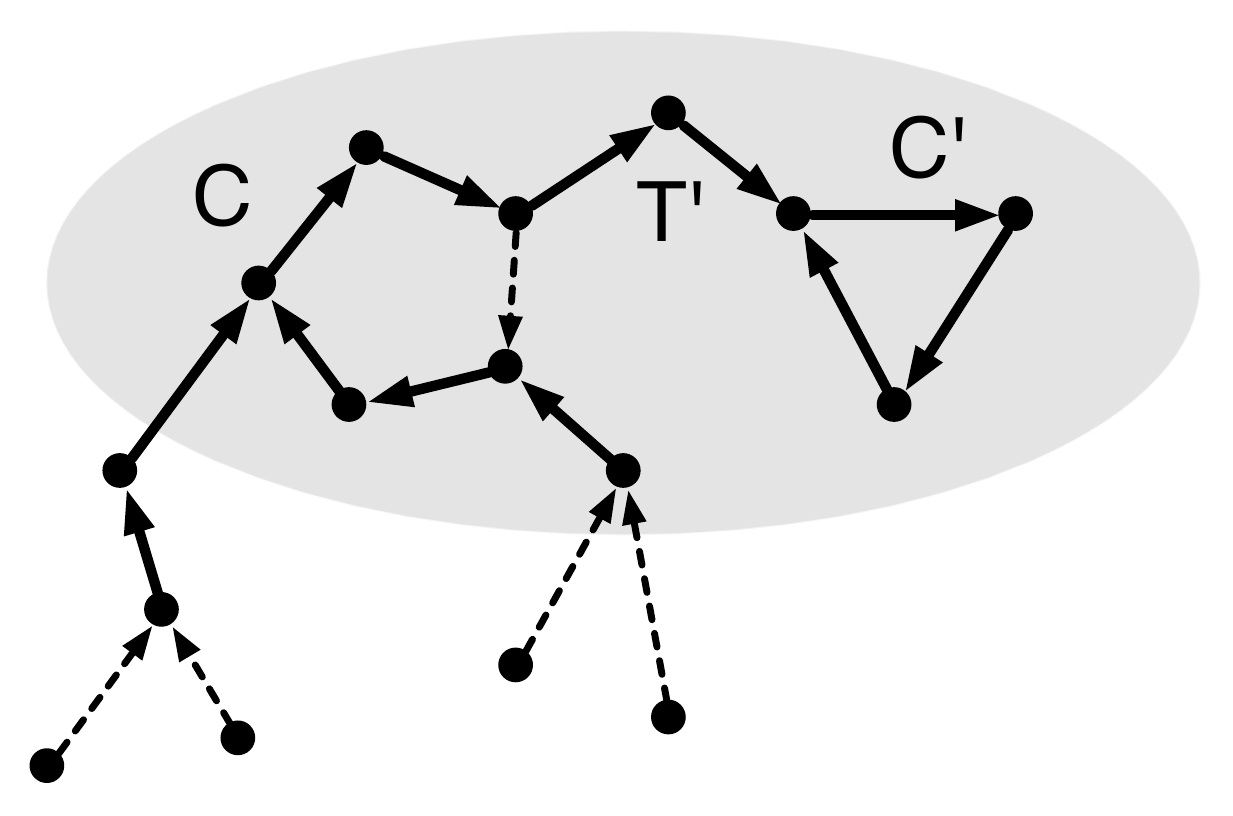}}
\caption{Construction of a witness containing the smallest cycle
of an SCC.  The shaded region highlights the vertices of an
SCC (edges not in $C$, $C'$, or $P$ are not depicted).  The edges of the witness are solid.  (a) The smallest
cycle $C'$ is not in the witness.  (b) By removing an edge from
$C$ and leaf edges from the in-arborescences rooted on $C$, we
create a witness that includes the smallest cycle $C'$.}
\label{fig:1-neighbour-structure}
\end{figure}

To describe the algorithm, let $\mathcal{D} = (S,F)$ be the DAG of maximal
SCCs of $G$ and let $\varepsilon > 1/k$ be a fixed constant where $k$ is
the knapsack bound. (If $\varepsilon \leq 1/k$
then the brute force algorithm which considers all subsets $V' \subseteq V(G)$
with $|V'| \leq k$ yields an acceptable bound for a PTAS.)

We say that $u\in S$ is {\em large} if $c(u) > \varepsilon\, k$,
{\em petite} if $1 < c(u) \leq \varepsilon\, k$, or {\em tiny} if $c(u)=1$.
Let $L$, $P$, and $T$ be the set of all large, petite and tiny SCCs
respectively.
Note that since $\varepsilon > 1/k$, for every $u \in L$, $c(u)> \varepsilon\, k >1$.

\begin{center}
\fbox{
\begin{minipage}[h]{.9 \linewidth}
\noindent {\sc uniform-directed-1-neighbour}
\begin{tabbing} \qquad $B = \emptyset$ \\
\qquad For every subset $X\subseteq L$ such that $|X| \le1/\varepsilon$\\
\qquad\qquad \= $D_X = \mathcal{D}[P \cup X]$.\\
\> $Z = \{ \mbox{tiny sinks of $\mathcal{D}$} \} \cup \{ \mbox{petite sinks of $D_X$} \}$ \\
\> $P' = $ any maximal subset of $Z$ such that $c(P') + c(X) \leq k$.\\
\> $U = \bigcup_{K \in P' \cup X}\{V(C)\ :\ C\mbox{ is a
smallest cycle of }K \}$\\
\> Greedily add vertices to $U$ such that $U$ remains a 1-neighbour \\
\> \qquad set until there are no more vertices to add or \\
\> \qquad $|U| = k$.  (Via a backwards search rooted at $U$.)\\
\> $B = \arg\max \{|B|, |U|\}$ \\
\qquad Return $B$.
\end{tabbing}
\end{minipage}
}
\end{center}

\begin{theorem} \label{thm:1-neighbour-ptas} {\sc
uniform-directed-1-neighbour} is a PTAS for the uniform, directed
1-neighbour knapsack problem.
\end{theorem}

\begin{proof}
Let $U^*$ be an optimal 1-neighbour knapsack and let $A_{U^*}$ be
its witness as guaranteed by Lemma~\ref{lem:scc-structure}.
Let $\mathcal{L}, \mathcal{P}$, and $\mathcal{T}$ be the sets of large, petite, and tiny cycles in
$A_{U^*}$ respectively.   
By Lemma~\ref{lem:scc-structure}, each of these cycles is in a
different maximal SCC and each cycle is a smallest cycle in its
maximal SCC.

Let $\mathcal{L}=\{L_{1}, \ldots, L_{\ell} \}$ and let $L^*$ be the set of large SCCs that intersect $L_1,\ldots,
L_\ell$.  Note that $|L^*| = \ell$.  Since $k \geq |U^*| \geq
\sum_{i=1}^\ell |L_{i}| > \ell\, \varepsilon\, k$ we have $\ell < 1/\varepsilon$.
So, in some iteration of {\sc uniform-directed-1-neighbour}, $X =
L^*$.  We analyze this iteration of the algorithm.  There are two
cases:
\begin{description}
\item[$P'=Z$.] First we show that every vertex in $U^*$ has a descendant in $X \cup
P'$.  Clearly if a vertex of $U^*$ has a descendant in some $L_i \in \mathcal{L}$, it has
a descendant in $X$.  Suppose a vertex of $U^*$ has a descendant in
some $P_i \in \mathcal{P}$.  $P_i$ is within an SCC of $D_X$, and so it must have a
descendant that is in a sink of $D_X$.  Similarly, suppose a vertex of $U^{*}$
has a descendant in some $T_{i} \in \mathcal{T}$.  $T_{i}$ is either a sink in $\mathcal{D}$ or has a
descendant that is either a sink of $\mathcal{D}$ or a sink of $D_{X}$.
All these sinks are contained in $X \cup P'$. Since every vertex of $U^*$ can reach a
vertex in $X \cup P'$, greedily adding to
this set results in $|U| = |U^*|$ and the result of {\sc
uniform-directed-1-neighbour} is optimal.

\item[$P' \neq Z$.]For any sink $x \notin P'$, $c(P')+c(X)+c(x) > k$ but $c(x)
\leq \varepsilon\, k$ by the definition of tiny and petite.  So, $|U| \geq
c(P')+c(X) > (1-\varepsilon) k$, and the resulting solution is within
$(1-\varepsilon)$ of optimal.
\end{description}

The running time of {\sc uniform-directed-1-neighbour} is
$n^{O(1/\varepsilon)}$.  It is dominated by the number of iterations,
each of which can be executed in poly time. \hfill  \end{proof}

\section{The uniform, undirected 1-neighbour  problem} \label{sec:uu1n}

We now consider the final case of 1-neighbour problems, namely the
uniform, undirected 1-neighbour problem.
We note that there is a relatively straightforward
linear time algorithm for finding an optimal solution for instances of this problem.  The algorithm essentially breaks the graph into connected components and then, using a counting argument, builds an optimal solution from the components.

\begin{theorem} \label{thm:uu1n}
The uniform, undirected 1-neighbour knapsack problem has a linear-time solution.
\end{theorem}

\begin{proof}
Let $\mathcal{G}=(\mathcal{G}_{1}, \mathcal{G}_{2}, \ldots,
\mathcal{G}_{t})$ be the connected components of the dependency graph
$G$ in {\em decreasing} order by size (we can find such an ordering in linear time).  Note that each connected
component $\mathcal{G}_{j}$ constitutes a feasible set for the
uniform, undirected 1-neighbour problem on $G$.  If $k$ is odd and
$|G_{j}|=2$ for all $j$, then the optimal solution has size $k-1$ since no vertex can
be included on its own.  In this case the first $\lfloor k/2 \rfloor$ connected
components constitutes a feasible, optimal solution.

Otherwise, let $i$ be smallest index such that $\sum_{j=1}^{i} |\mathcal{G}_{j}| > k$.  If $i=1$ then
let $\mathcal{S}=0$.  Otherwise, take $\mathcal{S}=\sum_{j=1}^{i-1}
|\mathcal{G}_{j}|$.  If $\mathcal{S}=k$ then the first $i-1$
components of $G$ have exactly $k$ nodes and constitute a feasible,
optimal solution for $G$.  Otherwise, by our choice of $i$,
$\mathcal{S}<k$ and $|\mathcal{G}_{i}|>k-\mathcal{S}$.  Let $U=(u_{1},u_{2}, \ldots, u_{|\mathcal{G}_{i}|})$ be an ordering of the nodes in
$\mathcal{G}_{i}$ given by a breadth-first search (start the search
from an arbitrary node).  Collect the first $k-\mathcal{S}$ nodes of $u$ in $U=\{u_{l} \,|\, l \leq
k-\mathcal{S}\}$.  We consider three cases:
\begin{enumerate}
\item If $|U|=1$ and $|\mathcal{G}_{t}|=1$, then the first $i-1$
connected components along with $\mathcal{G}_{t}$ constitute a
feasible, optimal solution.
\item If $|U|=1$ and $|\mathcal{G}_{t}|
\neq 1$, then $|\mathcal{G}_{1}|>2$.  If $k=1$ then return
$\emptyset$ since there is no feasible solution, otherwise drop an
appropriate node from $\mathcal{G}_{1}$ (one that keeps the rest of
$\mathcal{G}_{1}$ connected) and add $u_{2}$ to $U$ since
$|\mathcal{G}_{i}|>1$.  Now the first $i-1$ connected components
(without the one node in $\mathcal{G}_{1}$) along with $U$
constitute a feasible, optimal solution.
\item If $|U|>1$, then the
first $i-1$ connected components along with $U$ constitute a
feasible, optimal solution.
\end{enumerate}
\end{proof}

\section{The all-neighbours knapsack problem} \label{sec:all-neighbours}

In this section, we consider the all-neighbours knapsack problem.  Our
primary result is a PTAS for the uniform, directed all-neighbours
problem.  We also show that uniform, directed all-neighbours is
NP-hard in the strong sense, so no polynomial-time algorithm can yield
a better approximation unless P=NP.  In addition, we show that
uniform, undirected all-neighbours knapsack reduces to the classic
knapsack
problem.

A set of vertices $U$ is a {\em feasible} all-neighbours knapsack
solution if, for every vertex $u \in U$, $N_G(u) \subseteq U$.  Recall
that for an SCC $c \in V(\mathcal{D})$ is obtained by contracting
$V(c) \subseteq V(G)$.  For convenience, let $w(c) = w(V(c))$ and
$p(c) = p(V(c))$.  Let $\mathcal{S}=\{ \desc_\mathcal{D}(u) \,|\, u
\in V(\cal{D})\}$ be the set of descendant sets for every node of
$\mathcal{D}$.  We now show that all feasible solutions to the
all-neighbour knapsack problem can be decomposed into sets from
$\mathcal{S}$.

\begin{property} \label{prop:all-neighbours}
Every feasible solution to a general, directed all-neighbour instance has the form $\cup_{u \in Q} V(u)$ where $Q \subseteq \cal{S}$.
\end{property}

\begin{proof}
Let $U$ be a feasible solution for the dependency graph $G$.  We claim
that if $u \in U$ then there exists a set $S \in \mathcal{S}$ such
that $u \in V(S)$ and $V(S) \subseteq U$.  Notice that the
all-neighbours constraint implies that if $b$ is a neighbor of $a$ in
$G$ and $c$ is a neighbor of $b$ in $G$, then $a \in U$ implies $c \in
U$.  Thus, by transitivity, if $a \in U$ and $b$ is reachable from $a$
then $b \in U$.  Let $u \in U$ and $v$ be the node in $\mathcal{D}$
such that $u \in V(v)$.  Suppose that $w \in \desc_{\mathcal{D}}(v)$.
Then every node in $V(w)$ is reachable from $u$ in $G$ as is every
node in $V(\desc_{\mathcal{D}}(v))$ so $V(\desc_{\mathcal{D}}(v)
\subseteq U$ which proves the claim since $\desc_\mathcal{D}(v) \in
\mathcal{S}$.  The property follows.
\end{proof}

Property~\ref{prop:all-neighbours} tells us that if $U$ is a feasible solution for $G$ and $u \in U$, then every node reachable from $u$ in $G$ must also be in the optimal solution.  We use this property extensively throughout the rest of Section~\ref{sec:all-neighbours}.

\subsection{The uniform, directed all-neighbour knapsack problem}

We show that {\sc uniform-directed-all-neighbour} (below) is a
PTAS for the uniform, directed all-neighbours knapsack problem.  The
key ideas are to (a) identify a set $A$ of {\em heavy nodes} in $V(\mathcal{D})$ i.e., those nodes $v$ where $w(v)> \epsilon k$,
and then (b) augment subsets of the heavy nodes with nodes from
the set $B$ of {\em light nodes}, i.e., those nodes $v$ with $w(v)\leq
\epsilon k$.  We note that this algorithm works on the set of SCCs and
can handle the slightly more general than uniform case: that in which
the weight and profit of a vertex is equal, but different vertices may
have different weights.

\begin{center}
\fbox{
\begin{minipage}[h]{.9\linewidth}
\noindent {\sc uniform-directed-all-neighbour}
\begin{tabbing}
\qquad \=$A = \{v \in V(\mathcal{D}) \,|\, w(v) > \epsilon
k\}$, $B = S \setminus A$, $X = \emptyset$ \\
\> For every subset $A'$ of $A$ such that $|A'| \leq
1/\epsilon$ \\
\>  \qquad \=$T = \desc_\mathcal{D}(A')$ \\
\>  \> Let $B' = \{v \,|\,  v \in B \cap (V(\mathcal{D}) \setminus
T)\mbox{ and } N_{\mathcal{D}}(v) \subseteq T\}$ \\
\>  \> While $w(T) \leq k$ and $B' \neq \emptyset$ \\
\>  \> \qquad \= Add any element $b \in B'$ to $T$. \\
\>  \> \> Update $B' = \{v \,|\,  v \in B \cap (V(\mathcal{D})
\setminus T)\mbox{ and } N_\mathcal{D}(v) \subseteq T\}$ \\
\>  \> If $W(V(T)) > W(X)$ then $X = V(T)$ \\
\>  Return $X$\\
\end{tabbing}
\end{minipage}
}
\end{center}

\begin{theorem} \label{thm:uniform-directed-all}
{\sc uniform-directed-all-neighbour} is a PTAS for the uniform,
directed all-neighbour knapsack problem.
\end{theorem}

\begin{proof}
Let $U^*$ be a set of vertices of $G$ forming an optimal solution to
the uniform, directed all-neighbours knapsack problem.  By Property~\ref{prop:all-neighbours},
there is a subset of nodes $Q^{*} \subseteq \mathcal{D}$ such that $U^* =
\cup_{u\in Q^*} V(u)$. Let $A^* = U^* \cap A$.  Since the size of
any node in $A$ is at least $\epsilon k$ and the weight of $U^*$ is
at most $k$, $|A^*| \leq 1/\epsilon$.  Since all subsets of $A$ of
size at most $1/\epsilon$ are considered in the for loop of {\sc
uniform-directed-all-neighbours}, set $A^*$ will be one such set.

Let $D^* = \desc(A^*)$.  Let $\tilde{B}$ be all the nodes of $\mathcal{D}$ added to
the solution in all iterations of the while loop.  Let $T^*=D^* \cup \tilde{B}$.

Since $A^* \subseteq U^*$, $D^* \subseteq U^*$ by
Property~\ref{prop:all-neighbours}.  Let $B^* = U^* \setminus D^*$.  $\tilde B$ and $B^*$ are not necessarily the same
set of nodes.  Suppose $\tilde B$ and $B^*$ are not the same set of
nodes and $w(T^*) < (1-\epsilon)w(U^*)$.  Then there is a node $u
\in B^* \setminus \tilde B$ such that $u$'s neighbours are in
$T^*$.  Since $w(u) < \epsilon k$, $u$ could be added to $\tilde B$,
a contradiction.

We now bound the running time of {\sc uniform-directed-all-neighbour}.
Line 1, which find the set of heavy nodes $A\subseteq V(\mathcal{D})$,
compute a simple set difference, and initialize the return value, take
at most $O(n)$ time.  Since $|A| \leq \frac{n}{\epsilon k}$ and $|A'|
\leq 1/ \epsilon$ there are at most ${\frac{n}{ \epsilon k} \choose 1/
\epsilon} \leq (n / \epsilon k)^{1/\epsilon}$ subsets of $A$
considered in line 2, so line 2 executes at most $(n / \epsilon
k)^{1/\epsilon}$ times.  Since we will never execute line 4 more than
$n$ times we have an $O(n^{1+(1/\epsilon)})$-time algorithm.
\end{proof}

\begin{theorem} \label{thm:uniform-directed-all-hard}
The uniform, directed-all-neighbour problem is NP-hard.
\end{theorem}

\begin{proof}
We reduce the set-union knapsack problem to the uniform, directed
all-neighbours knapsack problem.  An instance of SUKP consists of a
base set of elements $S=\{x_1, x_2,\ldots, x_n\}$ where each $x_i$ has
an integer weight $w_i$, a positive integer capacity $c$, a target
profit $d$, a collection $C=\{S_1,S_2,\ldots, S_m\}$ where
$S_i\subseteq S$, each subset $S_i$ has a non-negative profit $p_i$.
Then the question asked is: Does there exist a sub-collection
$C'=\{S_{i_1}, S_{i_2},\ldots, S_{i_t}\}$ of $C$ such that
$\sum_{j=1}^t p_{i_j} \geq d$ and for $T=\cup_{j=1}^t S_{i_j}$,
$\sum_{x_s\in T} w_s \leq c$.  This problem is known to be NP-hard in
the strong sense even for the case where $w_i=p_i=1$ and $|S_i|=2$ for
$1\le i\le m$~\cite{goldschmidt-etal:nrl1994}.

We consider instances of SUKP where
every subset $S_j$ in $C$ has cardinality 2 and profit $p_j=1$.
Also, each element $x_i$ has weight $w_i=1$.
Let $c$ be the capacity and $d$ be the target profit.
Given such an instance of SUKP we define next an
instance of uniform, directed all-neighbours that has a solution if and only
if the SUKP instance has a solution.

Let $G=(V,A)$ be a directed graph where for each element $x_i$ there
is a strongly connected component $scc_i$ with $M=d+1$ nodes one
of which is labeled $z_i$.
Let $U_i$ denote the set of nodes in $scc_i$.
For each subset $S_j$ there is a node $v_j\in V$.
For every $x_i\in S_j$ there is an arc $(v_j,z_i)\in A$ and
these are the only other arcs.
Let $k=cM+d$ be the target party size.
Then we claim that there is a party of size $k$ if and only
if there is a solution to the SUKP instance
having weight at most $c$ and profit at least $d$.

Suppose there is solution $P$ of size $k$ to uniform, directed all-neighbours.
Since $k=cM+d$ and $M>d$, there must be some collection $K$ of
node sets $U_i$ of strongly connected components such
that $P$ contains the union of nodes of the $U_i$'s
in $K$ where $|K|\le c$.
Hence $P$ must also contain a set $Z$ of at least $d$ nodes $v_j$.
Since $P$ is feasible solution it must be that for every $v_j\in Z$
if $x_i\in S_j$ then $U_i \in P$.
It is straightforward then to check that the collection of sets
$C'=\{S_j\; : \; v_j \in Z \}$
is a solution to the SUKP instance
with profit $d\ge |Z|$ and
since $\cup_{v_j\in Z} S_j = \{ x_i \; : \; U_i\in K\}$)
it has weight at most $c$.

Now suppose $C'=\{S_{j_1},S_{j_2},\ldots, S_{j_t}\}$ is a solution
to the SUKP instance where $t\geq d$ and
$|\cup_{r=1}^t S_{j_r}| \le c$.
Let $N=\cup_{r=1}^t S_{j_r}$ and hence $|N|\leq c$.
Arbitrarily choose some $K\subseteq C'$ where $|K|=d$.
Then take $P'=\{v_j \,|\, S_j\in K\}$.
Let $N'$ be a set of elements such that $N\subseteq N'$ and $|N'|=c$.
Define $P''=\cup_{x_i \in N'} U_i$.
Since $K\subseteq C'$, it must be for every $v_j\in K$, if
$x_i\in S_j$ then $U_i \subseteq P''$.
Therefore $P=P'\bigcup P''$ is a solution to the all-neighbours
problem where $|P|=cM+d$.
\end{proof}

\subsection{The uniform, undirected all-neighbour knapsack problem}

The problem of uniform, undirected all-neighbour knapsack is
solvable in polynomial time.
In this case we just need to find the subset of connected components of $G$
whose total size is as large as possible without exceeding $k$.
But this is exactly the subset sum problem.  Since $k \leq n$, the standard
dynamic programming algorithm yields a truly polynomial-time $O(nk)$
solution.

\subsection{The general, all-neighbour knapsack problem}

As mentioned in Section~\ref{sec:related} the general, directed, all-neighbours
knapsack problem is a generalization of the partially ordered knapsack
problem~\cite{Kolliopoulos:2007p1242} which has been shown to be hard to
approximate within a $2^{\log^\delta n}$ factor unless
3SAT$\in$DTIME$(2^{n^{3/4+\epsilon}})$~\cite{Hajiaghayi:2006p1244}.
Hence the general, directed all-neighbours knapsack problem is hard
to approximate within this factor under the same complexity assumption.

In the undirected case, i.e., the case where the dependency graph $G$
is undirected, $\mathcal{D}$ becomes a set of disjoint nodes,
one for each connected component of $G$,
and $\mathcal{S}=V(\mathcal{D})$.  By Property~\ref{prop:all-neighbours},
we are left with the problem of finding a subset of nodes $Q \subseteq
V(\mathcal{D})$ such that $p(Q)$ is maximal subject to $w(Q) \leq k$.
But this is exactly the 0-1 knapsack problem which has a well-known
FPTAS.  Thus, general, undirected all-neighbours also has an FPTAS.
Contrast this with the uniform, directed all-neighbours problem.
There, the sets in $\mathcal{S}$ are not disjoint, so we cannot use
the 0-1 knapsack ideas.

\section{Future directions}

There are several open problems to consider, including closing gaps,
improving the running times of the PTASes, and giving approximation
algorithms for the general, directed versions of both 1-neighbour and
all-neighbour. We believe that fully understanding these problems will
lead to ideas for a much more general problem: maximizing a linear
function with a submodular constraint.

\paragraph{{\bf Acknowledgments}}
We thank Anupam Gupta for helpful
discussions in showing hardness of approximation for general,
directed 1-neighbour knapsack.

\bibliographystyle{plain}
\bibliography{party-problem}

\newpage

\end{document}